%% file: Main.tex
\documentclass[12pt, draftclsnofoot, onecolumn]{IEEEtran}

\input{PKG}

\usepackage{common}

\normalsize
\ifCLASSINFOpdf
\else
\fi

\begin{document}
\title{Identification Codes via Prime Numbers
}
\author{
	\vspace{0.3cm}
    \fontsize{11}{11} \selectfont \IEEEauthorblockN{Emad Zinoghli\IEEEauthorrefmark{1} and Mohammad Javad Salariseddigh\IEEEauthorrefmark{2}
    }
    \\
	\vspace{0.45cm}
    \fontsize{11}{11} \selectfont \IEEEauthorblockA{\IEEEauthorrefmark{1} Department of Electrical Engineering, Sharif University of Technology
    \\
    \fontsize{11}{11} \selectfont \IEEEauthorrefmark{2} Institute for Communications Engineering, Technical University of Munich
    }
    \\[.8em]
    Emails: emad.zinoghli@sharif.edu, mjss@tum.de
}

\maketitle

\begin{abstract}
We introduce a method for construction of identification codes based on prime number generation over the noiseless channels. The earliest method for such construction based on prime numbers was proposed by Ahlswede which relies on algorithms for generation of prime numbers. This method requires knowledge of $2^n$ first prime numbers for identification codes with block length $n,$ which is not computationally efficient. In this work, we revisit Ahlswede's scheme and propose a number of modifications. In particular, employing probabilistic prime generation algorithm, we guarantee that the prime keys generation is possible in polynomial time. Furthermore, additional improvements in terms of type II upper bound are derived and presented. Finally, we propose a method for identification coding based on hash functions which generalizes the Ahlswede's scheme.
\end{abstract}

\section{Introduction}
In the identification problem discussed by \cite{AD89}, effective communication strategies allow the receiver to accurately ascertain whether a specific message, relevant to a particular task, has been transmitted by the sender or not. This differs from Shannon's message transmission problem \cite{S48}, where the decoder aims to recover the original sent message. This problem has received growing interest in the investigation of different applications within the domain of post-Shannon and semantic goal-oriented communications, see \cite{Salariseddigh23_BSC_Future_Internet,Salariseddigh_PhD_Diss} for extensive discussions.

In a randomized identification (RI) identification \cite{AD89}, the encoder utilizes randomness to generate codewords for messages. Employing randomness in the encoding procedure enables the decoder to identify double exponential number of messages, i.e., $\sim 2^{2^{nR}},$ where $n$ is the block length and $R$ indicates the coding rate. This is a remarkable improvement compared to the conventional message transmission problem of Shannon \cite{S48} or deterministic identification over the binary symmetric channel (BSC) \cite{Salariseddigh23_BSC_GC23,Salariseddigh23_BSC_Future_Internet} which are capable of reliable decoding for only exponential number of messages.

Construction of RI identification codes based on Reed-Solomon codes is studied in \cite{Verdu89}. One of the first construction of an identification code was proposed by Ahlswede and Verboven in \cite{Ahlswede91} titled the 3-step scheme. Based on this scheme, two random prime numbers that are distributed uniformly in a given interval, are generated. Then, these prime numbers are used to encode the message. Therein, it is shown that such encoding method achieves the channel capacity for a noiseless binary channels. In this work, we revisit this method and show that utilizing probabilistic prime number generating algorithm can be beneficial in terms of computational complexity. We consider identification systems over the binary symmetric channel that are interested to accomplish the identification task, namely, to verify whether or not a particular message has been sent at the transmitter. In particular, we make the following contributions:
\begin{itemize}
\item \textbf{\textcolor{mycolor12}{Time Complexity:}}
To generate uniform primes on the interval $[1:p_K]$ where $p_n$ is the $n_{\,\text{th}}$ prime number, Ahlswede and Verboven \cite{Ahlswede91} pick a random index $k \gets \{1, \ldots, K\}$ and then calculate $p_k.$ Although this method minimizes the number of random bits and generates exactly uniform primes, it is computationally expensive. Currently, the best algorithms compute the first $K$ primes in sub-exponential time in the size of the input, the number of bits of $K.$ For example, the Atkin's sieve \cite{Atkin04} runs in $\mathcal{O}(K / \log \log K).$ Even if we restrict ourselves to only computing $p_k,$ the best algorithms are still sub-exponential. In \cite[Chapter 9.9]{bach}, it is shown that $p_k$ can be computed efficiently given an oracle that compute $\pi(x)$ and vice versa. The current best algorithm for computing $\pi(x)$ runs in $\mathcal{O}(x^{1/2 + \varepsilon})$ \cite{lagarias} for any $\varepsilon > 0$ which is still sub-exponential in the number of bits of $x.$ Our main contribution is a way of improving the time complexity of the 3-step scheme. We consider polynomial time algorithms that generates \textit{almost uniform} primes with the least number of random bits possible. We employ the Miller-Rabin test to generate such prime numbers.

\item \textbf{\textcolor{mycolor12}{Error Bound:}}
In our modified identification scheme we propose sending the prime numbers as identifier instead of their index. This eliminates the need for a sub-exponential algorithm to calculate the index of a prime in the transmitter and another sub-exponential algorithm to calculate prime numbers from their index in the receiver. Thus, design of transmitter and receiver systems may be made more efficient and simpler. Although this technique increases the size of block length, we show that we are still able to identify doubly exponential number of messages by providing tighter upper bounds for the type II error probability.

\item \textbf{\textcolor{mycolor12}{Generalized Coding Scheme:}}
We establish a connection between the RI codes and universal hash functions, through which a systematic construction for RI codes is enabled. In this method we exploit both the Shannon transmission codes and universal hash functions to construct RI codes.
\end{itemize}

\subsection{Notations}
We use the following notations throughout this paper:
Blackboard bold letters $\mathbbmss{K,X,Y,Z},\ldots$ are used for alphabet sets. Lower case letters $x,y,z,\ldots$ stand for constants and values of random variables, and upper case letters $X,Y,Z,\ldots$ stand for random variables. The set of consecutive natural numbers from $1$ through $M$ is denoted by $[\![M]\!].$ Suppose \(\mathcal{D}\) is a distribution over the set \(S\), by \(x \gets \mathcal{D}\) we mean that \(x\) is chosen from the set \(S\) according to the distribution \(\mathcal{D}\). When \(x\) is chosen uniformly we simply denote it as \(x \gets S\). The notation $f(n) = \mathcal{O}(g(n))$ is used to indicate that function $f(n)$ is asymptotically dominated by function $g(n).$ We denote \(f \sim g\) when \(f\) and \(g\) are asymptotically the same, that is $\lim_{x\to \infty} f(x)/g(x) = 1.$ The notation $\circ$ stands for composition of functions. Throughout this report we donate base \(2\) logarithm by \(\log\) and natural logarithm by \(\ln\).

\subsection{Organization}
Section~\ref{Sec.SysModel} includes required preliminaries on the identification codes, and system model. Section~\ref{Sec.Prelim_Num_Theory} provides fundamental number theoretic concepts and results. In Section~\ref{Sec.PNG}, we provide probabilistic algorithm for prime generation. Section~\ref{Sec.3Step–ID} review the previous 3–step identification scheme in details. Section~\ref{Sec.3Step–ID–Modified} presents our main contribution, i.e., an improved 3-step coding scheme. Section~\ref{Sec.3Step–ID–Generalized} includes a generalization of the previous encoding scheme. Then, in Section~\ref{Sec.Simulations} we provide numerical experiments and relevant analysis. Finally, we conclude this work in Section~\ref{Sec.Conclusions} with summary, discussion, and future research directions.

\section{System Model and Preliminaries}
\label{Sec.SysModel}
In this section, we present the adopted system model and establish some preliminaries regarding identification coding for the BSC.

\subsection{System Model}
We address an identification-focused communication setup, where the decoder's purpose is accomplishing the following task: Determining whether or not a target message has been sent by the transmitter.

\subsection{Identification Coding for the Noiseless Binary Symmetric Channel}
The definition of an RI code for the BSC is given below.
\begin{definition}[Identification Code for Noiseless Channels]
\label{Def.BSC-DI-Code}
An $(n,\allowbreak M(n,R), \allowbreak 0, \allowbreak \lambda_2)$ identification code for a BSC $\bB$ with channel matrix $W$ for integer $M(n,R)$ where $n$ and $R$ are the block length and coding rate, respectively, is defined as a system $(\C,\D),$ which consists of a codebook $\C = \{ \mathbf{c}_i \},$ for every $i \in [\![M]\!]$ with $\fc_i = (c_{i,t})|_{t=1}^n \subset \{0,1\}^n,$ and a collection of decoders $$\D = \bigcup_{j \in [\![M]\!]} \D_j \,,$$ where $\D_j \subset \{0,1\}^n$ is the decoding set corresponding to the single message $\fc_j.$
Given a message $i \in [\![M]\!],$ the encoder transmits codeword $\mathbf{c}_i,$ and the decoder's asks: Was a target message $j \in [\![M]\!]$ sent or not? There exist two errors that may happen:
\begin{align*}
    P_{e,1}(i) = 1 -\sum_{\fy \in \D_i} W^n \big( \fy \, | \, \fc^i \big) \qquad \text{ and } \qquad
    P_{e,2}(i,j) = \sum_{\fy \in \D_j} W^n \big( \fy \, | \, \fc^i \big) .
\end{align*}
It must hold $P_{e,1}(i) = 0$ and $P_{e,2}(i,j) \leq \lambda_2, \forall \, i,j \in [\![M]\!]$ with $i \neq j, \, \forall \lambda_2 > 0.$ A rate $R > 0$ is achievable if $\forall \lambda_2 > 0$ and sufficiently large $n,$ there exists an identification code meeting the error conditions. The operational identification capacity for $\bB,$ denoted by $\mathbb{C}_{\rm I}(\bB),$ is supremum of all achievable rates.
\end{definition}

\section{Preliminaries on Number Theory and Hash Functions}
\label{Sec.Prelim_Num_Theory}
In this section, we establish some elementary definitions and theorems on prime numbers and their asymptotic distribution amongst positive integers. These preliminaries support our understanding of 3-step algorithm which exploits the prime number generation for code construction.
\begin{definition}
    \label{Def.Prime-Num}
    A prime number (prime integer or a \emph{prime} for short) is a positive integer $p > 1$ that has no positive integer divisors other than $1$ and $p$ itself. In other word, a prime number $p$ is a positive integer which has exactly one positive divisor except than $1 \,,$ which implies that a prime $p$ is a number that cannot be factored.
\end{definition}
\begin{definition}
    \label{Def.Prime-Counting-Func}
    Assume that $x > 0$ is a real number. Then, the prime-counting function $\pi(x),$ is a function that counts the number of primes not exceeding $x.$ That is,
    \begin{align}
        \label{Eq.Prime-Counting-Func}
        \pi(x) = \left| \left\{ p \leq x \;:\; p \text{ is a prime } \right\} \right|
        \,.
    \end{align}
\end{definition}
    Observe that by Euclid's Theorem, there are infinitely many prime numbers amongst the positive integers. Therefore, we immediately conclude that the function $\pi(x)$ provided in \ref{Eq.Prime-Counting-Func} diverges to infinity as $x$ tends to infinity, i.e., $\pi(x) \to \infty$ as $x \to \infty \,.$ The behaviour (growth rate) of the function $\pi(x)$ has been subject to immense studies by numerous mathematicians. In particular, there has been several observations which support the intuitive idea that the prime numbers become less common as they become larger. In order to formalize such an idea by quantifying precisely the rate at which this phenomenon (becoming less common) occurs, a fundamental theorem called \emph{Prime Number Theorem (PNT)} is stated which describes the asymptotic distribution of the prime numbers among the positive integers.

\begin{definition}
    Let $p_n$ be the $n$-th prime number, then analog to the standard factorial for prime numbers, the primorial $p_n \#$ is defined as the product of the first $n$ primes, i.e.,
    \begin{align}
        p_n\# = \prod_{k=1}^n p_k
        \,.
    \end{align}
\end{definition}
\begin{theorem}[{{\cite[Ch.~4]{Apostol76}}}]
\label{Th.PNT}
Let $\pi(x)$ denote the prime-counting function for real $x > 0 \,.$ Then,
\begin{align}
    \lim_{x \to \infty} \frac{\pi(x)}{x / \ln x} = 1
    \;.
\end{align}
\end{theorem}
\begin{proof}
    The proof is provided in \cite[Ch.~4]{Apostol76}.
\end{proof}

Theorem~\ref{Th.PNT} is referred to as the PNT and is equivalent to the statement that the $n$-th prime number $p_n$ satisfies $p_n \sim n \log n.$ Furthermore, PNT suggests that the density of primes in natural number $\mathbb{N}$ is zero, i.e.,
\begin{align}
    \lim_{n \to \infty} \frac{\pi(n)}{n} = \frac{1}{\ln n} = 0 .
\end{align}

In the following, we introduce a theorem which connects the PNT and the asymptotic value of the $n$-th prime number.
\begin{theorem}[{{\cite[Th.~4.5]{Apostol76}}}]
    \label{Th.n-th-Prime}
    Let $p_n$ indicate the $n$-th prime number. Then the following asymptotic relations are equivalent
    \begin{align}
        \lim_{x \to \infty} \frac{\pi(x)}{x / \ln x} = 1 \,, \quad \lim_{x \to \infty} \frac{\pi(x)}{x / \ln \pi(x)} = 1 \,, \quad \lim_{n \to \infty} \frac{p_n}{n \ln n} = 1 .
    \end{align}
\end{theorem}
\begin{proof}
    The proof is provided in \cite[Ch.~4]{Apostol76}.
\end{proof}
\subsection{Non-Asymptotic Bounds for the Prime-Counting Function and the $n$-th Prime Number}
Theorem~\ref{Th.PNT} and Theorem~\ref{Th.n-th-Prime} provide asymptotic behaviour for the prime-counting function, $\pi(x)$ and the $n$-th prime number, $p_n.$ In particular, the PNT is an asymptotic result which gives an \emph{ineffective} bound on the prime-counting function, $\pi(x) \,,$ as a direct consequence of the definition of the limit, i.e., $\forall \varepsilon > 0 \,,\;
\exists N \in \mathbb{N}$ such that $\forall n \geq N \,,$ we have
\begin{align}
   (1 - \varepsilon) \frac{x}{\ln x} < \pi(x) < (1 + \varepsilon) \frac{x}{\ln x}
   \,.
\end{align}
However, in some analysis, we may require exact analytic lower and upper bounds on $\pi(x)$ and $p_n \,.$ Therefore, we introduce the following theorems which establish exact bounds on such functions.
\begin{theorem}[{{\cite[Th.~4.6]{Apostol76}}}]
\label{Th.Prime-Counting-Func-Bounds}
Assume that $n \geq 2$ is a positive integer and let $\pi(n)$ denote the prime-counting function associated to $n \,.$ Then, $\pi(x)$ is bounded by
\begin{align}
    \frac{1}{6} \Big( \dfrac{n}{\ln n} \Big) < \pi(n) < 6 \Big( \frac{n}{\ln n} \Big)
    \,.
\end{align}
\end{theorem}
\begin{proof}
    The proof is provided in \cite[P.~82, Ch.~4]{Apostol76}.
\end{proof}
\begin{theorem}[{{\cite[Th.~4.7]{Apostol76}}}]
\label{Th.n-th-Prime-Bounds}
Assume that $n \geq 1$ is a positive integer and let $p_n$ denote the $n$-th prime number. Then, $p_n$ is bounded by
\begin{align}
    \frac{1}{6} n \ln n \leq p_n \leq 12 \big( n \ln n + n \ln ( 12 / e ) \big)
    \,.
\end{align}
\end{theorem}
\begin{proof}
    The proof is provided in \cite[P.~84, Ch.~4]{Apostol76}.
\end{proof}
In order to establish an upper bound on the type II error probability of the 3-step scheme proposed by Ahlswede in \cite{Ahlswede91}, it is argued that the number of distinct prime factors of a number $n$ is less than $\log n \,.$ We utilize the following theorem which establishes a tighter upper bound for such a function.
\begin{theorem}[{{\cite[Sec.~22.10]{Hardy79}}}]
\label{Th.omega-n-Order}
Let the prime factorization of a natural number $n$ be $n = p_1^{a_1} \ldots p_r^{a_r} \,.$ Further, let $\omega(n) = r$ be the number of distinct prime factors of $n \,.$ Then, for primorial $n,$ i.e., $n = \prod_{k=1}^r p_k = p_r \# \,,$ we obtain 
\begin{align}
	\omega(n) \sim \frac{\ln n}{\ln \ln n}
    \,,
\end{align}
which is equivalent to the following
\begin{align}
	\omega(n) = \frac{\ln n}{\ln \ln n} + o(\frac{\ln n}{\ln \ln n})
    \quad \equiv \quad
    \lim_{n \to \infty} \frac{\omega(n)}{\ln n / \ln \ln n} = 1
    \,.
\end{align}
\end{theorem}

\begin{theorem}[\cite{Ramanujan17}]
\label{Th.omega-n-Order-UB}
	Let \(\func{\omega}{n}\) denote the number of distinct prime factors of \(n\).  For all \(\varepsilon > 0\) and sufficiently large values of \(n\)
\begin{equation}
	\func{\omega}{n}  \leq \dfrac{\ln n}{\ln \ln n} + \varepsilon .
\end{equation}
\end{theorem}



\begin{definition}[\cite{Hardy79}]
    Suppose that $P$ is a property of a positive integer, and $P(x)$ is the number of numbers less than $x$ possessing the property $P,$ i.e.
    \begin{align}
        P(x) = | \{ n \,,\; n \leq x \;, n \text{ has the property $P$} \} | .
        
    \end{align}
    Now, if $P(x) \sim x \,,$ when $x \to \infty \,,$ then the number of numbers less than $x$ which do not posses the property $P$ is $o(x) \,.$ Then, we say that almost all numbers posses the property $P \,.$
\end{definition}
\begin{definition}[\cite{carter}]
    Let \(H = \set{h : X \to Y}\) be a family of hash functions from \(X\) to \(Y\) and let \(\varepsilon\) be a positive real number. \(H\) is said to be \(\varepsilon\)-almost universal if for all distinct \(x_1, x_2 \in X\) 
    \begin{equation}
        \abs{\set{h \in H \middle| \func{h}{x_1} = \func{h}{x_2}}} \leq \varepsilon \abs{H}
    \end{equation}
    If \(h \gets H\) uniformly, then $\prob{\func{h}{x_1} = \func{h}{x_2}} \leq \varepsilon.$
\end{definition}

\section{Prime Number Generation}
\label{Sec.PNG}
In this section, we establish some algorithms which generate prime numbers.

\subsection{Uniform Prime Generation}
\label{subsec:unifprime}
The 3-step scheme picks its prime keys from a uniform distribution. To implement this method we need algorithms that can compute $k_{\rm th}$ prime from $k.$ As stated before, the current best algorithms run in sub-exponential time in the number of bits of $k \,.$ Therefore, we need to relax some of these conditions to obtain practical algorithms. A trivial algorithm for producing uniform primes is given in Algorithm~\ref{Alg.UPM-1}.

\vspace{3mm}
\input{Alg-1}
\vspace{3mm}

When we use a deterministic primality test in Algorithm~\ref{Alg.UPM-1}, the distribution of primes is exactly uniform. This algorithm may never terminate, however, we expect it to stop after $\mathcal{O}(\log n)$ steps. Because exploiting Theorem~\ref{Th.PNT}, we have
\begin{align}
    \frac{\pi(n)}{n} \sim \frac{1}{\ln n}
    \,.
\end{align}
Hence, on average in $\mathcal{O}(\log n)$ steps a prime number $p$ is chosen. As a result, Algorithm~\ref{Alg.UPM-1} uses an average of $\mathcal{O}(\log^2 n)$ random bits. The current state-of-the-art deterministic primality tests, Agrawal–Kayal–Saxena primality test (AKS), runs in $\log ( \mathcal{O}(\log^6 n) )$ \cite{Lenstra19,Agrawal04} which means that on average Algorithm~\ref{Alg.UPM-1} terminates in $\log ( \mathcal{O}(\log^7 n) ).$

We can further improve the time complexity of Algorithm~\ref{Alg.UPM-1} if we use randomized primality tests. These tests can determine whether a number $p$ is prime with high probability.

\vspace{3mm}
\input{Alg-2}
\vspace{3mm}

For example, the Miller-Rabin test might declare a composite number as a prime, however, the probability of this event cab made arbitrary small. The output of the Algorithm~\ref{Alg.UPM-2} is not a uniform prime number as it can be composite, however, the distribution of prime numbers is equi-probable over all primes less than $n \,.$ Each round of the Miller-Rabin test  uses $\log p$ random bits where $p$ is the number that is to be tested. Therefore, we still use an average of $\mathcal{O}(\log^2 n)$ random bits. The test itself runs in $\mathcal{O}(\log^3 n)$ \cite{bach} thus, the Algorithm~\ref{Alg.UPM-2} terminates in $\mathcal{O}(\log^4 n) \,.$

In this work, we implement the Miller-Rabin test since it is more efficient and easier to implement. Furthermore, by executing this test an appropriate number of rounds, we can ensure that the resulting distribution is statistically close to the uniform distribution over primes.

\subsection{Miller-Rabin Analysis}
Miller-Rabin is a well-known random primality test algorithm which is easy to implement, simple, and fast. Let $MR(n,k)$ be the distribution of Miller-Rabin algorithm on the prime candidate $n$ where $k$ denotes the number of rounds of test that are performed. Let $\P$ be the set of primes, then using \cite[Th.~9.4.5]{bach}, we obtain
\begin{align}
    \Pr(MR(n,k) = 1 \,|\, n \in \P ) & = 1
    \nonumber\\
	\Pr(MR(n,k) = 1 \,|\, n \notin \P ) & \leq 4^{-k} .
\end{align}
Consider the following random prime number generator, $GNR(N,s,k)$, as described in Algorithm~\ref{alg:GMR}. This algorithm, samples numbers uniformly and then checks if they are prime using the Miller-Rabin test. The parameter $N$ is the upper bound, $s$ is the maximum number of samples, and $k$ is the number of repeats in the underlying Miller-Rabin test.

\vspace{3mm}
\input{Alg-3}
\vspace{3mm}

An analysis of the distribution of $GMR(N,s,k)$ is given as follows. Let \(n_i\) denote the random variable \(n\) in the \(i_{\,\text{th}}\) iteration.
\begin{align}
	\prob{\func{GMR}{M,s,k} = \perp} & = \prob{\func{MR}{n_1,k} = \dots =  \func{MR}{n_s,k} = 0}\\
	&= \prod_{i = 1}^s \prob{\func{MR}{n_i,k} = 0} & \text{(Independence)} \nonumber \\
	&= \prod_{i = 1}^s \condProb{\func{MR}{n_i,k} = 0}{n_i \notin \P} \prob{n_i \notin \P}\\
	&\leq \prod_{i = 1}^s \bracket{1 - \dfrac{\func{\pi}{N}}{N}}\\
	&=\bracket{1 - \dfrac{\func{\pi}{N}}{N}}^s \\
	&\leq \bracket{1 - \dfrac{1}{6 \ln N}}^s & \text{(Theorem~\ref{Th.PNT})} . \nonumber
\end{align}

If we bound this error probability with \(\varepsilon\), then exploiting Taylor series approximation of $\ln(1-x)$ for sufficiently large $N,$ we obtain the following bound on \(s\).
\begin{equation}
	s \geq \frac{\ln \varepsilon}{\ln(1 - \frac{1}{6\ln N})} \approx - 6\ln N \ln \varepsilon = -\frac{6}{ (\log e)^2} \log N \log \varepsilon .
\end{equation}
The probability that the result of $GMR(N,s,k)$ is composite, given that it is not $\perp$ is as follows.

\begin{align}
	& \condProb{\func{GMR}{N,s,k} \notin \P}{\func{GMR}{N,s,k} \neq \perp}
    \nonumber\\&
    \leq \sum_{i = 1}^s \prob{\func{MR}{n_i, k} = 1, n_i \notin \P } &&& \text{(Union bound)}
    \nonumber\\&
    = \sum_{i = 1}^s \condProb{\func{MR}{n_i, k} = 1}{ n_i \notin \P }\prob{n_i \notin \P}
    \nonumber\\&
    \leq \sum_{i = 1}^s 4^{-k} \bracket{1 - \dfrac{\func{\pi}{N}}{N}}
    \nonumber\\&
    = s 4^{-k}\bracket{1 - \dfrac{\func{\pi}{N}}{N}}
    \nonumber\\&
    \leq s4^{-k}\bracket{1 - \dfrac{1}{6\ln N}}  &&& \text{(Theorem~\ref{Th.PNT})} .
\end{align}
If we bound this error probability with $\delta$ then we get the following bound on $s:$
\begin{equation}
	s \leq 1 - (\frac{1}{6\ln N})^{-1}  4^k \delta \approx   4^{k}\delta ,
\end{equation}
for sufficiently large \(N\). Let \(\varepsilon = 2^{-l}\) and \(\delta = 2^{-q}\) with \(l,q \geq 0\). Then, 
\begin{equation}
    \label{eq:skparameters}
	\dfrac{6}{(\log e)^2} l \log N \leq 3  l \log N  \leq s \leq 2^{2k- q} .
\end{equation}
Note that, setting \(s = 3l\log N\) and \(k = \dfrac{{\log}{3l} + \log \log N + q}{2}\) satisfies both inequalities.

\section{Previous Results - 3-Step Identification Scheme for the Noiseless BSC}
\label{Sec.3Step–ID}
In this subsection, we present the original results proposed in \cite{Ahlswede91} for a noiseless BSC.

Assume that set of messages is indicated by $\M =  [\![M]\!] \,,$ and let $\alpha > 1$ is a fixed constant. Then, let define
\begin{align}
    K = \ceil{(\log M)^{\alpha}} ,
\end{align}
and $\pi_1 < \pi_2 < \ldots < \pi_{K}$ as the consecutive set of $K$ smallest prime numbers. Further, for $k \in \K \triangleq [\![K]\!]$ define a key $\varphi_k: \M \to [\![\pi_k]\!]$ as follows
\begin{align}
    \varphi_k(m) - 1 & \equiv m \pmod {\pi_k}
    \,.
\end{align}
Then, let $\{ \varphi_k \}_{k \in \K}$ be a cipher and $\M' = \{ \varphi_k(m) \}_{m \in \M \,, k \in \K}$ the set of all possible enciphering serving as a message set for a second cipher denoted by $\{ \varphi'_l \}_{L \in \K'}$ where $\K' \triangleq [\![K']\!]$ with 
\begin{align}
    K' = \ceil{(\log \pi_k)^{\alpha}}
    \,,
\end{align}
and the second cipher $\varphi_l: \M' \to [\![\pi_l]\!]$ satisfies
\begin{align}
    \varphi'_l(m') - 1 & \equiv m' \pmod {\pi_l}
    \,.
\end{align}
Next, we proceed to state the 3-step for encoding procedure.
\begin{itemize}
    \item \textbf{\textcolor{mycolor12}{Step 1}}:
    The sender chooses $k \in \K$ randomly according to the uniform distribution on the set $\K$ and transmits it (and also the key $\varphi_k$) over the channel. This requires $\ceil{\log K}$ bits.
    \item \textbf{\textcolor{mycolor12}{Step 2}}:
    Similarly the sender selects an $l \in \K'$ at random and sends it (and also the key $\varphi'_l$) over the channel. This requires $\ceil{\log K'}$ bits.
    \item \textbf{\textcolor{mycolor12}{Step 3}}:
    Let $m \in \M$ be given to the sender for transmitting over the channel and assume that the received message in the receiver is $\hat{m}.$ It calculates $\varphi'_l(\varphi_k(m))$ and sends it to the receiver. Given the fact that the receiver knows both $k$ and $l,$ the receiver calculates $\varphi'_l(\varphi_k(\hat{m}))$ and compare it with the transmitted encryption $\varphi'_l(\varphi_k(m)) \,.$ Then, the decoder makes a decision as follows
     
    $$\begin{cases}
        m = \hat{m} & \varphi'_l(\varphi_k(\hat{m})) = \varphi'_l(\varphi_k(m)) \,,
        \\
        m \neq \hat{m} & \varphi'_l(\varphi_k(\hat{m})) \neq \varphi'_l(\varphi_k(m)) \,.
    \end{cases}$$
\end{itemize}
\vspace{3mm}

\begin{theorem}[Optimality of the 3-Step Scheme,{{\cite[Sec.~22.10]{Hardy79}}}]
    The above 3-step identification scheme for the noiseless BSC is \emph{optimal}. That is, the following three properties hold:
    \begin{enumerate}
        \item The error probability of type I equals zero. 
        \item The error probability of type II tends to zero as the block length $n$ tends to infinity.
        \item There exist a codebook whose size fulfill the following:
        \begin{align}
            \lim_{n \to \infty} \frac{\log \log M(n)}{n} = \frac{1}{\alpha} .
        \end{align}
        where $\alpha > 1$ is an arbitrary constant.
    \end{enumerate}
\end{theorem}
For the type II error probability analysis we present the following lemma.

\begin{lemma}
		Any positive integer \(m\) has at most \(\floor{\log m}\) unique prime factors. 
	\end{lemma}
	\begin{proof}
		Suppose \(q_1, \dots, q_k\) are all the prime factors of \(m\). Then for some \(\alpha_1, \dots, \alpha_k \geq 1\)
		\begin{equation*}
			m = \prod_{i = 1}^k q_i^{\alpha_i} \geq \prod_{i = 1}^{k} 2^{\alpha_i} \geq 2^k .
		\end{equation*}
		As a result, \(k \leq \floor{\log m}\) as required.
	\end{proof}
 
	\begin{lemma}
		For any \(m,m' \in \mathcal{M} = \set{1,2, \dots,M}\) such that \(m \neq \hat{m}\)
		\begin{equation}
			\abs{ \set{k  \in \set{1,2, \dots, K}\;\middle|\;\func{\phi_k}{m} = \func{\phi_k}{\hat{m}}}} \leq \log M .
		\end{equation}
	\end{lemma}
 
	\begin{proof}
		The given set consists of common prime factors of \(m\) and \(\hat{m}\) that are less than or equal to \(p_K\). The inequality immediately follows from the fact that \(m,\hat{m} \leq M\) and \(M\) has at most \(\log M\) prime factors.
	\end{proof}
	We can derive an upper bound for the second kind error.
	\begin{align}
		\func{P_{e,2}}{m,\hat{m}} &= \condProb{\func{\phi_l}{\func{\phi_k}{m}} = \func{\phi_l}{\func{\phi_k}{\hat{m}}}}{ m \neq \hat{m}}\\
		&=\condProb{\func{\phi_l}{\func{\phi_k}{m}} = \func{\phi_l}{\func{\phi_k}{\hat{m}}}}{\func{\phi_k}{m} = \func{\phi_k}{\hat{m}}} \condProb{\func{\phi_k}{m} = \func{\phi_k}{\hat{m}}}{ m \neq \hat{m}} \nonumber\\
		& \quad +\condProb{\func{\phi_l}{\func{\phi_k}{m}} = \func{\phi_l}{\func{\phi_k}{\hat{m}}}}{\func{\phi_k}{m} \neq \func{\phi_k}{\hat{m}}} \condProb{\func{\phi_k}{m} \neq \func{\phi_k}{\hat{m}}}{ m \neq \hat{m}}
        \nonumber\\
		&\leq \condProb{\func{\phi_k}{m} = \func{\phi_k}{\hat{m}}}{ m \neq \hat{m}} + \condProb{\func{\phi_l}{\func{\phi_k}{m}} = \func{\phi_l}{\func{\phi_k}{\hat{m}}}}{\func{\phi_k}{m} \neq \func{\phi_k}{\hat{m}}}
        \nonumber\\
		&\leq \dfrac{\log M}{K} + \dfrac{\log M'}{K'}
        \nonumber\\
		&= \dfrac{\log M}{\ceil{\bracket{\log M}^{\alpha}}} + \dfrac{\log p_K}{\ceil{\bracket{\log p_K}^{\alpha}}}
        \nonumber\\
		&\leq \dfrac{1}{\bracket{\log M}^{\alpha - 1}} + \dfrac{1}{\bracket{\log p_K}^{\alpha - 1}} .
	\end{align}
	By the prime number theorem, Theorem~\ref{Th.PNT}, \(p_K \sim K \ln K\). As a result, \(\lambda_2 \to 0\) as \(M \to \infty\).
	\begin{align}
		(\log p_K)^{\alpha - 1} &\sim (\log K + \log \log K - \log \log e)^{\alpha - 1}
        \nonumber\\&
        \approx (\alpha \log \log M + \log\log\log M + \log \alpha - \log\log e)^{\alpha - 1} .
	\end{align}
Finally, the block length is calculated by \(n = \ceil{\log K} + \ceil{\log K'} + \ceil{\log p_{K'}}\). By applying the prime number Theorem~\ref{Th.PNT}, we obtain
	\begin{align}
		n &= \ceil{\log K}  + \ceil{\log K'} + \ceil{\log p_{K'}}\\
		&= \ceil{\log \ceil{(\log M)^{\alpha}}} + \ceil{\log \ceil{(\log p_K)^{\alpha}}} + \ceil{\log p_{K'}}\\
		&\approx \alpha \log \log M + \alpha \log\log p_K + \log p_{K'}\\
		&\approx \alpha \log \log M + (1 + o(1))\log\log\log M + (\alpha + o(1)) \log \log \pi_K \\
		&\approx \alpha (1 + o(1)) \log\log M .
	\end{align}
This scheme requires both sender and receiver to have access to a prime generation algorithm that given \(k\) computes the \(k_{\mathrm{th}}\) prime, \(p_k\). As we have mentioned earlier, this problem does not have polynomial algorithm yet. To alleviate these inefficiencies we propose the following modifications.

\section{Main Results -Part I - Modified 3-Step Identification Scheme}
\label{Sec.3Step–ID–Modified}
In this section, we present a modified version of the original 3-step identification scheme provided in \cite{Ahlswede91}. In particular, consider the following conventions:
\begin{itemize}
    	\item Let \(\mathcal{M} = \set{1,2, \dots, M }\) be the message set and let $$K = \ceil{\bracket{\log M}^{\alpha}} \qquad \text{ and } \qquad K' = \ceil{\bracket{\log K}^{\alpha}},$$  for some constant \(\alpha > 1\). 
	\item Let us denote the set \(\set{1,2, \dots, l }\) by \(\mathbb{Z}^+_l\).  Define the function \(\phi_l: \mathbb{N} \to \mathbb{Z}_{l}^+ \) as follows.
	\begin{equation}
		\func{\phi_l}{n}= [n \mod l] + 1 ,
	\end{equation}
	where \([n \mod l]\) is equal to the remainder of the division of \(n\) by \(l\).
\end{itemize}

A round of communication in this scheme is as follows.
\begin{enumerate}
	\item The sender chooses a probabilistic prime \(k\) from the set \(\mathcal{K} = \set{1,2,\dots, K}\) by a prime number generator and transmits it.
	\item The sender chooses another probabilistic prime \(l\) from the set \( \mathcal{K}' = \set{1,2,\dots, K'}\) by the same prime number generator and transmits it.
	\item Given a message \(m \in \mathcal{M}\), the sender transmits \(\func{\phi_l}{\func{\phi_k}{m}}\). Assuming that receiver wishes to identify \(\hat{m} \in \mathcal{M}\), he calculates \(\func{\phi_l}{\func{\phi_k}{\hat{m}}}\) and compares it with \(\func{\phi_l}{\func{\phi_k}{m}}\). He identifies the message as \(\hat{m}\) whenever \(\func{\phi_l}{\func{\phi_k}{m}}\) = \(\func{\phi_l}{\func{\phi_k}{\hat{m}}}\).
\end{enumerate}

\begin{theorem}[Optimality of the Modified 3-Step Scheme]
    \label{thm:modified3step}
	Our proposed modified 3-step identification coding scheme for the noiseless BSC is \emph{optimal}. That is, the following three properties hold:
    \begin{enumerate}
        \item The error probability of type I equals zero. 
        \item The error probability of type II tends to zero as the block length $n$ tends to infinity.
        \item There exist an identification codebook whose size fulfill the following:
        \begin{align}
            \label{Eq.Optimality_Modified}
            \lim_{n \to \infty} \frac{\log \log M(n)}{n} = \frac{1}{\alpha} .
        \end{align}
        where $\alpha > 1$ is an arbitrary constant. That is, any rate arbitrary close to 1 is achievable.
    \end{enumerate}
\end{theorem}

\begin{proof}
The proof consists of rate and error analysis and are provided in Subsections~ \ref{Subsec.RateAnalysis} and \ref{Subsec.ErrorAnalysis}.
\end{proof}

\subsection{Rate Analysis}
\label{Subsec.RateAnalysis}
The block length is calculated by \(n = \ceil{\log K} + \ceil{\log K'} + \ceil{\log l}\), i.e.,
\begin{align}
	n &= \ceil{\log K} + \ceil{\log K'} + \ceil{\log l}
    \nonumber\\&
    \approx \ceil{\log K} + 2\ceil{\log K'}
    \nonumber\\&
    \approx \alpha\log\log M +  2 \alpha \log \log K
    \nonumber\\&
    \approx \alpha\log \log M + 2\alpha \log \log \log M + 2 \alpha \log \alpha
	\nonumber\\&
    \approx \alpha(1 + \littleO{1}) \log\log M .
\end{align}
Therefore, the limit given in \eqref{Eq.Optimality_Modified} is proved.
\subsection{Error Analysis}
\label{Subsec.ErrorAnalysis}
The error analysis of our modified coding scheme depends on the prime number generator. We use a simple prime number generator Algorithm~\ref{alg:GMR} based on the Miller-Rabin primality test. In order to calculate the type II error probability, we provide the following useful lemma.
\begin{lemma}\label{lmm:modified}
    Suppose \(m,\hat{m} \in \mathcal{M} = \set{1,2, \dots, M}\) and \(p \gets \func{\mathcal{A}}{K,\varepsilon}\) is a probabilistic prime generated by randomized algorithm \(\mathcal{A}\) such that all primes less than or equal to \(K\) are generated equi-probable and 
    \begin{equation}
        \prob{p \notin \mathcal{P}} \leq \varepsilon .
    \end{equation}
    Then, 
    \begin{equation}
        \condProb{\func{\phi_p}{m} = \func{\phi_p}{\hat{m}}}{ m\neq \hat{m} } \leq \dfrac{\log M}{\func{\pi}{K} \log \log M } + \varepsilon .
    \end{equation}
\end{lemma}

\begin{proof}
From Theorem~\ref{Th.omega-n-Order-UB}, for all \(\delta > 0\) and sufficiently large \(n\),
\begin{equation}
	\func{\omega}{n}\leq \dfrac{\ln n}{\ln \ln n} +\delta .
\end{equation}
Then, by the previous argument
    \begin{equation}
        \dfrac{\abs{\set{p \leq K \;\middle| \;\func{\phi_{\pi}}{m} = \func{\phi_{\pi}}{\hat{m}}} }}{\abs{\set{p \leq K}}} \leq \dfrac{\func{\omega}{\abs{m - \hat{m}}}}{\func{\pi}{K}} \leq \dfrac{\log M}{\func{\pi}{K} \log \log M} .
    \end{equation}
    Therefore,
\begin{align}
    \condProb{\func{\phi_p}{m} = \func{\phi_p}{\hat{m}}}{ m\neq \hat{m} } & = \condProb{\func{\phi_p}{m} = \func{\phi_p}{\hat{m}}}{ p \in \mathcal{P}, m\neq \hat{m} }\prob{p \in \mathcal{P}}
    \nonumber\\&
    + \condProb{\func{\phi_p}{m} = \func{\phi_p}{\hat{m}}}{ p \notin \mathcal{P}, m\neq \hat{m} }\prob{p \notin \mathcal{P}}\\
    &\leq \condProb{\func{\phi_p}{m} = \func{\phi_p}{\hat{m}}}{ p \in \mathcal{P}, m\neq \hat{m} } + \varepsilon\\
    &\leq \dfrac{\log M}{\func{\pi}{K} \log \log M} + \varepsilon .
\end{align}
\end{proof}

\begin{proof}
	Let \(m,\hat{m} \in \mathcal{M} = \set{1,2,\dots, M}\) be given such that \(m \neq \hat{m}\). Let \(\mathcal{A}\) be the randomized prime generator as described above. Suppose the probable primes \(k \gets \func{\mathcal{A}}{K,\varepsilon}\) and \(l \gets \func{\mathcal{A}}{K',\varepsilon}\) are generated with \(K = \ceil{\bracket{\log M}^{\alpha}}\) and \(K' =\ceil{\bracket{\log K}^{\alpha}}\). 
	The probability of the second kind error is given as follow 
	\begin{align}
		\func{P_{e,2}}{m,\hat{m}} &= \condProb{\func{\phi_l}{\func{\phi_k}{m}} = \func{\phi_l}{\func{\phi_k}{\hat{m}}} }{m \neq \hat{m}}\\
		 &=  \condProb{\func{\phi_l}{\func{\phi_k}{m}} = \func{\phi_l}{\func{\phi_k}{\hat{m}}} }{ \func{\phi_k}{m} = \func{\phi_k}{\hat{m}}, m \neq \hat{m}} \nonumber\\ 
         &\qquad \condProb{ \func{\phi_k}{m} = \func{\phi_k}{\hat{m}} }{m \neq \hat{m}}\nonumber\\
		 & \qquad + \condProb{\func{\phi_l}{\func{\phi_k}{m}} = \func{\phi_l}{\func{\phi_k}{\hat{m}}} }{ \func{\phi_k}{m} \neq \func{\phi_k}{\hat{m}}, m \neq \hat{m}} \nonumber\\
         &\qquad \condProb{ \func{\phi_k}{m} \neq  \func{\phi_k}{\hat{m}} }{m \neq \hat{m}}\\
		 &\leq \condProb{ \func{\phi_k}{m} = \func{\phi_k}{\hat{m}} }{m \neq \hat{m}} \nonumber\\
        & \qquad + \condProb{\func{\phi_l}{\func{\phi_k}{m}} = \func{\phi_l}{\func{\phi_k}{\hat{m}}} }{ \func{\phi_k}{m} \neq \func{\phi_k}{\hat{m}}} \\
		 &\leq \dfrac{\log M}{\func{\pi}{K} \log\log M} + \dfrac{\log K}{\func{\pi}{K'} \log\log K} + 2\varepsilon && \text{(Lemma~\ref{lmm:modified})} \nonumber \\
		 &\leq \dfrac{6\func{\log}{M} \func{\log}{K}}{ K \log\log M  }+ \dfrac{6\func{\log}{K} \func{\log}{K'}}{ K' \log\log K }  + 2\varepsilon && \text{(Theorem~\ref{Th.PNT})} \nonumber\\
		 &\approx \dfrac{6 \alpha \log \log M}{ \bracket{\log M}^{\alpha-1} \log\log M }+ \dfrac{6 \alpha \log \log K}{ \bracket{\log K}^{\alpha-1} \log\log K }   + 2\varepsilon\\
		 &= \dfrac{6 \alpha}{ \bracket{\log M}^{\alpha-1}}+ \dfrac{6 \alpha }{ \bracket{\log K}^{\alpha-1}} +2\varepsilon .
	\end{align}
Thereby, since \(\varepsilon\) can be chosen as small as possible, then \(\lambda_2 \to 0\) as \(M \to \infty\). 
This completes the proof of Theorem~\ref{thm:modified3step}.
\end{proof}

\section{Main Results -Part II - Generalized 3-Step Algorithm}
\label{Sec.3Step–ID–Generalized}
In this section, we introduce the 3-step algorithm in terms of universal hash functions. Moreover, we establish an equivalence between universal hash functions and identification codes. This equivalence allows us to design new identification systematically and then compare them all in one framework.

In the 3-step algorithm, we generate two random prime numbers and then compute a function of the message based on these two primes. In general, this is viewed as producing a local randomness \(j\) and computing a message specific tag function \(T_i\) with input \(j\), \(\func{T_i}{j}\). When working with a noiseless channel, the tagging function determines the second kind error probability. For example, the probability that the receiver identifies \(\hat{i}\) when \(i\) is sent is calculated as
\begin{equation}
    P_{e,2}(i,\hat{i}) = \prob{\func{T_i}{j} = \func{T_{\hat{i}}}{j}}.
\end{equation}
When \(j\) is chosen uniformly, this probability can be viewed as the probability of the collision of the hash function \(\func{h_j}{x} = \func{T_x}{j}\). The families of hash functions for which this probability of collision is bounded are called \textit{almost universal}.

Next, we introduce an identification scheme for noiseless channels that utilizes almost universal functions.

\begin{definition}\label{def:hash3step}
    Let \(\mathcal{M} = \set{1 , \dots,  M}\) and \(H = \set{h_{a} : X \to Y}_{a \in [I]}\) be a family of \(\varepsilon\)-almost universal hash functions indexed by the set \([I] = \set{0, 1 , \dots , I - 1}\). A round of communication round works as follows.
    \begin{enumerate}
        \item The sender chooses \(a \gets [I]\) uniformly. Then, he sends the code \((a, \func{h_a}{m})\) to the receiver.
        \item Upon receiving the code \((a, \func{h_a}{m})\), the receiver calculates \(\func{h_a}{\hat{m}}\) and identifies \(\hat{m}\) whenever \(\func{h_a}{m} = \func{h_a}{\hat{m}}\). 
    \end{enumerate}  
\end{definition}
The block length of the identification scheme in Definition~\ref{def:hash3step} is $n = \log I + \log \abs{Y}.$ The number of random bits required is $r = \log I.$ Then, the second kind error probability reads
\begin{equation}
    \condProb{\func{h_a}{m} = \func{h_a}{\hat{m}}}{m \neq \hat{m}} \leq \varepsilon .
\end{equation}
Thus, we have shown the following statement.
\begin{theorem}
    If there exists a class of \(\varepsilon\)-almost universal hash functions \(H = \set{h_a : \mathbb{F}_q^n \to \mathbb{F}_q}_{a \in [I]}\) then there exists a \((\log I + \log q, q^n,0, \varepsilon )\) identification code for noiseless binary channel.
\end{theorem}

\subsection{Formulation of 3-step algorithm in terms of universal hash function}

 We now give an explicit construction of the Definition~\ref{def:hash3step} according to the Section~\ref{Sec.3Step–ID–Modified}. Let \(\func{A}{n} = \set{1,2, \dots, 2^n }\), \(\func{B}{n} = \set{1,2, \dots ,n^{\alpha} }\) for some \(\alpha > 1\), \(K = \set{2, 3, \dots, p_{\func{\pi}{n^{\alpha}}}}\), and define hash family \(\func{H}{n} = \set{h_k : A \to B \middle| k \in K}\) given as follows.
 \begin{equation}
    \func{h_p}{a} = \squareBracket{a \mod p} + 1 .
 \end{equation}
 \begin{lemma}\label{lmm:hashuniv}
    The class \(\func{H}{n}\) described above is \(\frac{\alpha}{n^{\alpha - 1}}\)-almost universal.
 \end{lemma}

 \begin{proof}
    Suppose \(x , y \in A\) are distinct. 
    \begin{align}
        \prob{\func{h_p}{x} = \func{h_p}{y}} &= \dfrac{1}{\func{\pi}{n^{\alpha}}} \abs{\set{p \in K \middle| \ p \mid  \abs{x-y}}} \leq \dfrac{\func{\omega}{x - y}}{\func{\pi}{n^{\alpha}}} \approx \dfrac{\alpha}{n^{\alpha - 1}} .
    \end{align}
 \end{proof}
Note that, \(H\) digests input exponentially, therefore, if we use \(H\) twice we can achieve double exponential compression. Firstly, consider the following composition lemma.
 \begin{lemma}\label{lmm:hashcomp}
    Suppose \(H_1: A \to B\) and \(H_2: B \to C\) are \(\varepsilon_1\) and \(\varepsilon_2\)-almost universal, respectively. The hash family \(H = H_2 \circ H_1 : A \to C = \set{h_2 \circ h_1 \middle| h_1 \in H_1, h_2 \in H_2}\) is \(\varepsilon =( \varepsilon_1 + \varepsilon_2)\)-almost universal. It is shown that \(H\) is \((\varepsilon_1 + \varepsilon_2 - \varepsilon_1\varepsilon_2)\)-almost universal \cite{bierbrauer}.
 \end{lemma}

\begin{proof}
    Let \(x, y \in A\) be two distinct elements.
    \begin{align}
        & \prob{\func{h_2}{\func{h_1}{x}} = \func{h_2}{\func{h_1}{y}}} 
        \nonumber\\&
        \leq \condProb{\func{h_2}{\func{h_1}{x}} = \func{h_2}{\func{h_1}{y}}}{\func{h_1}{x} \neq \func{h_1}{y}} +  \prob{\func{h_1}{x} = \func{h_1}{y}}
        \nonumber\\& \leq \varepsilon_2 + \varepsilon_1 .
    \end{align}
\end{proof}
By Lemma~\ref{lmm:hashuniv} and Lemma~\ref{lmm:hashcomp}, the hash family \(\func{H^2}{n} = \func{H}{\alpha n} \circ \func{H}{2^n}\) is $\varepsilon$-almost universal where $$\varepsilon = \frac{\alpha}{(\alpha n)^{\alpha - 1}} + \frac{\alpha}{2^{n \bracket{\alpha - 1}}} = \mathcal{O}(n^{1 - \alpha})$$ and maps \(A = \set{0, 1, \dots, 2^{2^n}}\) to \(C = \set{1, \dots, (\alpha n)^{\alpha}}\).
\subsection{Current Directions}
One way to construct identification codes is by re-purposing error correcting codes \cite{Verdu93}. The methods described in \cite{Verdu93} use the error correction codes to construct a tag function. The work \cite{bierbrauer} has shown the equivalence between error correcting codes and almost universal hash function. 

\begin{theorem}\label{thm:hasheqid}
    If there is a \([n,k,d]_q\) code, then there exists a \((1- \frac{d}{n})\)-almost universal hash family \(H = \set{h_a : \mathbb{F}_q^k \to \mathbb{F}_q}_{a \in [n]}\) and conversely, if there is a \(\varepsilon\)-almost universal  \(H = \set{h_a :  \mathbb{F}_q^k \to \mathbb{F}_q}_{a \in [n]}\) then there exists a \([n, k,n(1- \varepsilon)]_q\) code.
\end{theorem}

\section{Performance Analysis and Simulation Experiments}
\label{Sec.Simulations}
In this section, we discuss on the complexity of our algorithm and some existing trade-offs between the performance parameters. Finally, we introduce our simulations experiments along and describe the relevant analysis.

\subsection{Theoretical Complexity}
As oppose to the original 3-step algorithm of Ahlswede, in our revision, most of the work is done by the transmitter. Specifically, the transmitter runs the prime number generator \(\mathcal{A}\) twice. Suppose the simple Algorithm~\ref{alg:GMR} is used with Miller-Rabin primality test. By the analysis given in the Section~\ref{subsec:unifprime}, this algorithm on average terminates in $\mathcal{O}((\log n)^4)$ where \(n\) is its input. Therefore, on input \(K\), the algorithm runs in 

\begin{equation}
    \mathcal{O}((\log K)^4) = \mathcal{O}(\alpha \log \log M )^4) = \mathcal{O}(n^4) ,
\end{equation}

where \(n\) is the block length. On input \(K'\), it runs in 

\begin{equation}
     \mathcal{O}((\log K')^4) = \mathcal{O}(\alpha \log \log K )^4) = \mathcal{O}(\log n^4) .
\end{equation}

All together, we have significantly improved time complexity of the scheme without affecting its error probability or block length. We should mention that this is the time complexity to generate prime numbers. To compute the actual code, we need to divide a \(2^n\)-bit message by an \(n\)-bit prime. However, this is common in all identification schemes as we are required to process exponentially larges messages. 

Note that, we can fix the encoding procedure such that the probable prime keys \(k\) and \(l\) are padded with zero so they are sent with exactly \(\ceil{\log K}\) and \(\ceil{\log K'}\) bits respectively. This alleviates the problem of separating the keys. Moreover, since the size of the keys can be easily computed from the scheme parameters, \(M\) and \(\alpha\), this technique does not require additional complicated circuit to implement. In conclusion, the transmitter requires a circuit to generate primes and another circuit to compute the divisions. Meanwhile, the receiver only requires the division circuit.

\subsection{Trade-offs}
The complexity of this scheme for a given number of messages is tuned with the parameter \(\alpha\). Increasing \(\alpha\), will improve the second kind error probability significantly but comes at expense of  block length and time. However as evident from Equations~\cref{eq:tradeoffs1,eq:tradeoffs2,eq:tradeoffs3}, the effect of changing \(\alpha\) is different for block length, error probability, and time.
\begin{align}
    n &\approx \alpha \log \log M \label{eq:tradeoffs1}\\
    P_{2,e} &\approx \dfrac{6 \alpha}{(\alpha \log \log M)^{\alpha - 1}}\label{eq:tradeoffs2}\\
    \text{Time} &= \bigO{(\alpha \log \log M)^4} .
    \label{eq:tradeoffs3}
\end{align}

Thus, setting the value of \(\alpha\) depends on the resources available and the required performance.

\subsection{Numerical Experiments}
We have implemented a code that mainly runs the Algorithm~\ref{Alg.UPM-1} to generate two probable prime numbers. Then it calculates the block length, error probability, the probability of collision, and the running time of the simulation. The time in simulations is measured in milliseconds. The parameters \(s,k\) are set according to Equation~\eqref{eq:skparameters} with \(l = q = 10\). The main parameters are \(\alpha\), \(M\), and \(R\) the number of rounds that we repeat the simulation. We have selected \(\log \log M\) as the independent value that runs from \(100\) to \(2000\) by increments of \(100\). We then plotted the block length, second kind error probability, the probability of collision, and time against \(\log \log M\) for different values of \(\alpha\) and \(R\). 

We define the probability of collision for a pair of prime keys \(k,l\) as 
\begin{equation}
    P_{coll,k,l} = \dfrac{\abs{\set{ \set{x,y} \in \mathcal{M} \times \mathcal{M} \middle| \phi_l(\phi_k(x)) = \phi_l(\phi_k(y)), x \neq y }}}{M^2} .
\end{equation}
That is, the probability that two message have the same identifier. Since, \(l\) is usually smaller than \(k\), then we can approximate the probability collision as follows.
\begin{equation}
    P_{coll,k,l} \approx \dfrac{\abs{\set{ \set{x,y} \in \mathcal{M} \times \mathcal{M} \middle| \phi_l(x) = \phi_l(y), x \neq y }}}{M^2} \approx \dfrac{1}{l} \approx \bigO{(\alpha \log \log M)^{-\alpha}} .
\end{equation}
By trial and error, we have found that $12 (\alpha \log \log M)^{-\alpha}$ matches the curve the best. We shall note that the probability of collision and the second kind error are different entities. Observe that $\mathrm{Time} \approx 0.018 (\alpha \log \log M)^4$ matches the obtained curves.

\textbf{Discussions:}
As our figures suggest, increasing \(\alpha\) causes the second kind error probability to decrease much more rapidly in expense of a larger block length. The parameter \(R\) determines the number of times that we run the simulation. Increasing \(R\) will cause the the average probability of collision to be more accurate. Similarly, decreasing \(R\) will causes the average  probability of collision to deviate from its correct average. This explains the random spikes and dips in the obtained figures.

\begin{figure}[H]
    \centering
    \setkeys{Gin}{width=0.5\textwidth}
    \subfloat[\label{fig:subfig-a}]{\includegraphics{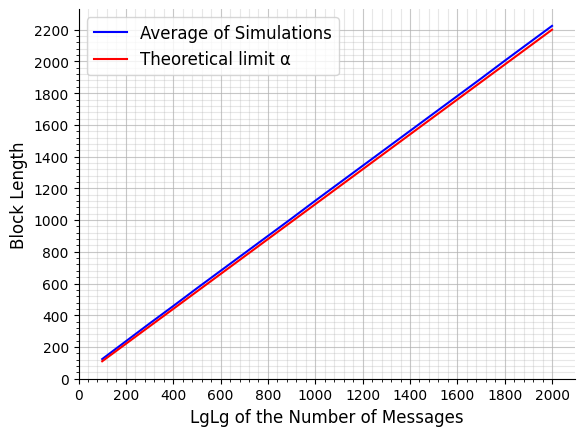}}
    \hfill
    \subfloat[\label{fig:subfig-b}]{\includegraphics{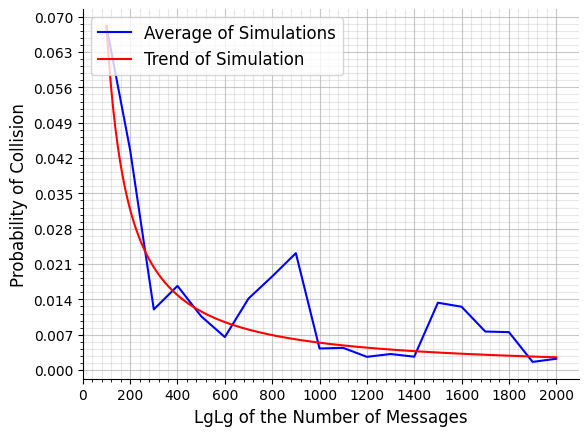}}
    \hfill
    \subfloat[\label{fig:subfig-c}]{\includegraphics{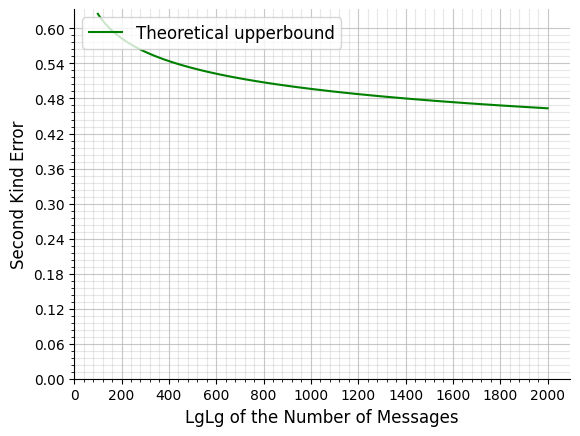}}
    \hfill
    \subfloat[\label{fig:subfig-d}]{\includegraphics{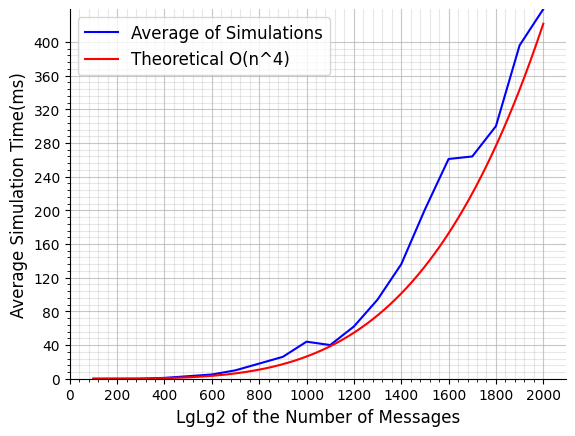}}
    \caption{Theoretical upper bound and empirical type II error probability corresponding to parameters \(\alpha = 1.1\) and \(R = 100\)}
\end{figure}
\begin{figure}[H]
    \centering
    \setkeys{Gin}{width=0.5\textwidth}
    \subfloat[\label{fig:subfig-e}]{\includegraphics{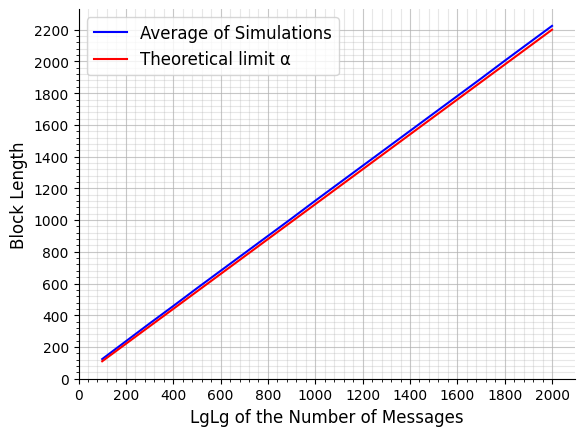}}
    \hfill
    \subfloat[\label{fig:subfig-f}]{\includegraphics{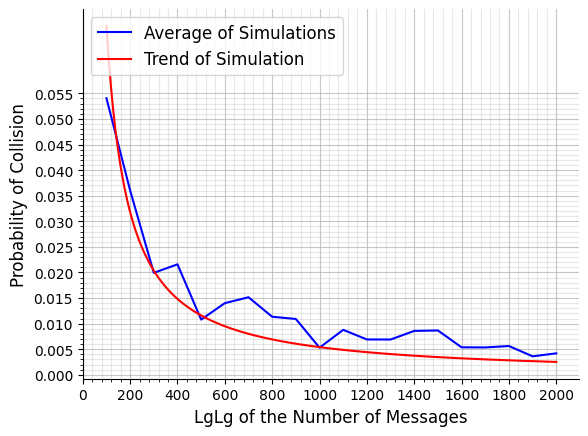}}
    \hfill
    \subfloat[\label{fig:subfig-g}]{\includegraphics{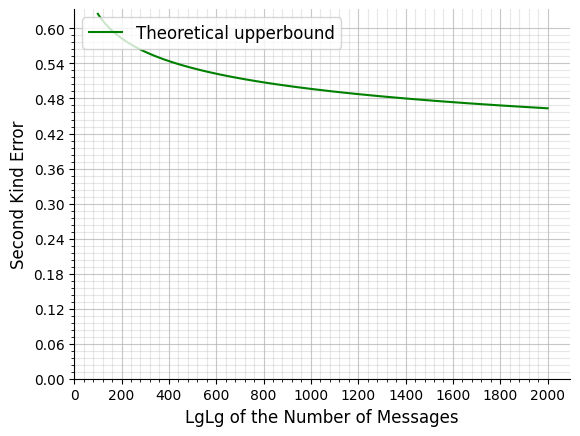}}
    \hfill
    \subfloat[\label{fig:subfig-h}]{\includegraphics{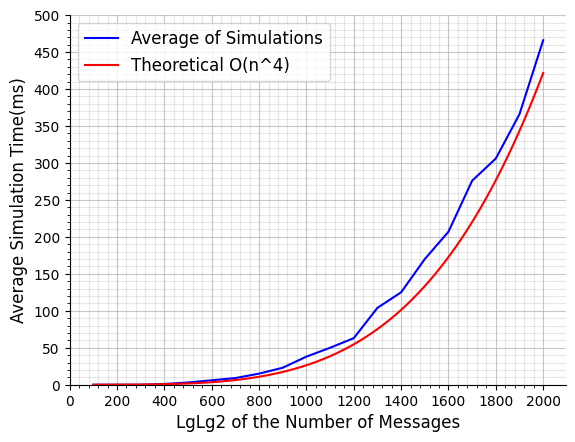}}
    \caption{Theoretical upper bound and empirical type II error probability corresponding to parameters \(\alpha = 1.1\) and \(R = 500\)}
\end{figure}

\begin{figure}[H]
    \centering
    \setkeys{Gin}{width=0.5\textwidth}
    \subfloat[\label{fig:subfig-i}]{\includegraphics{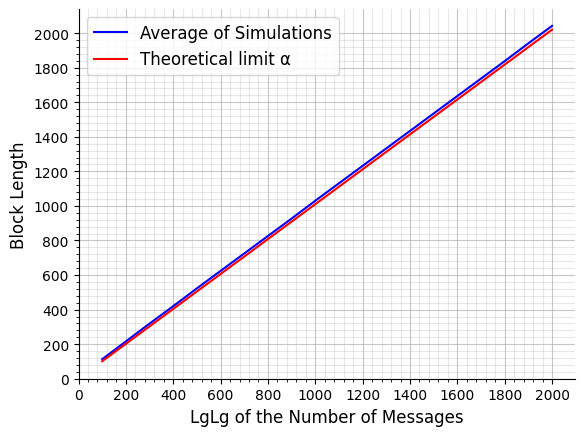}}
    \hfill
    \subfloat[\label{fig:subfig-j}]{\includegraphics{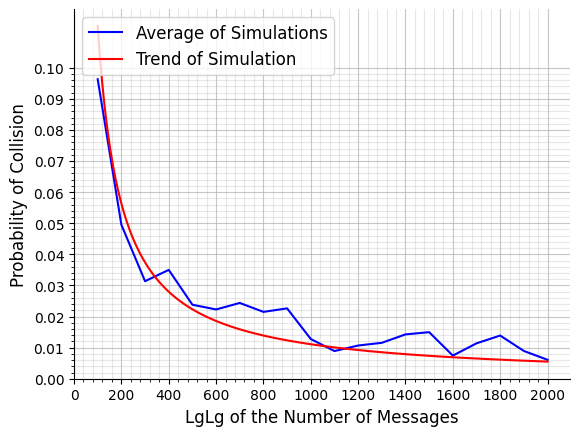}}
    \hfill
    \subfloat[\label{fig:subfig-k}]{\includegraphics{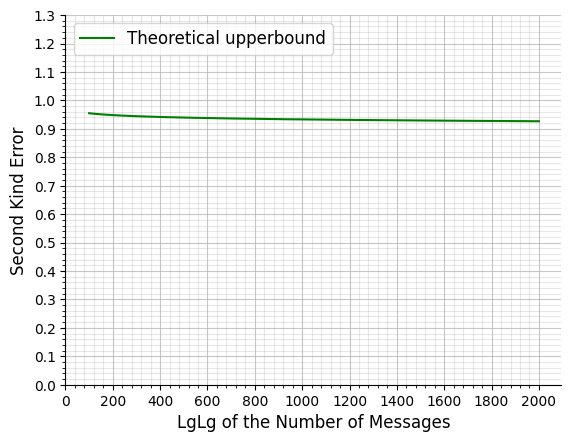}}
    \hfill
    \subfloat[\label{fig:subfig-l}]{\includegraphics{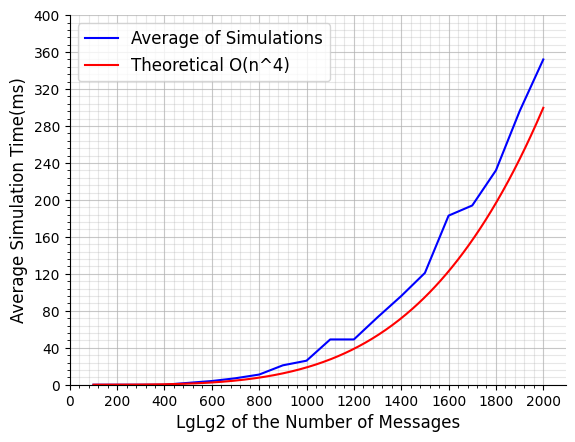}}
    \caption{Theoretical upper bound and empirical type II error probability corresponding to parameters \(\alpha = 1.01\) and \(R = 500\)}
\end{figure}

\begin{figure}[H]
    \centering
    \setkeys{Gin}{width=0.5\textwidth}
    \subfloat[\label{fig:subfig-m}]{\includegraphics{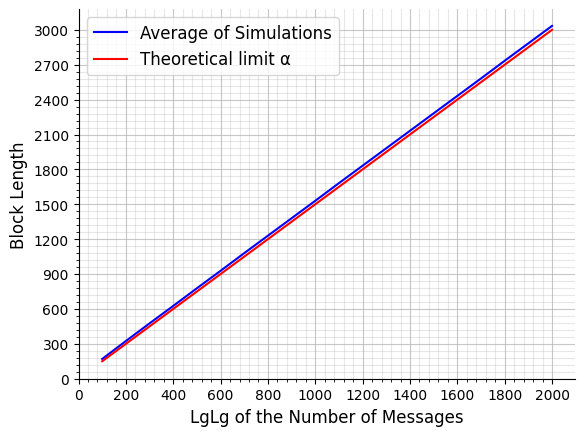}}
    \hfill
    \subfloat[\label{fig:subfig-n}]{\includegraphics{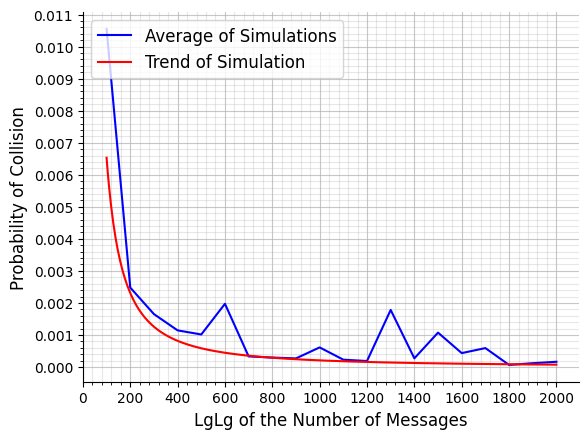}}
    \hfill
    \subfloat[\label{fig:subfig-o}]{\includegraphics{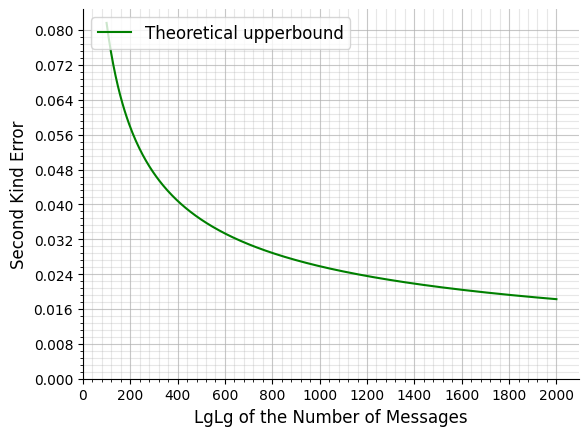}}
    \hfill
    \subfloat[\label{fig:subfig-p}]{\includegraphics{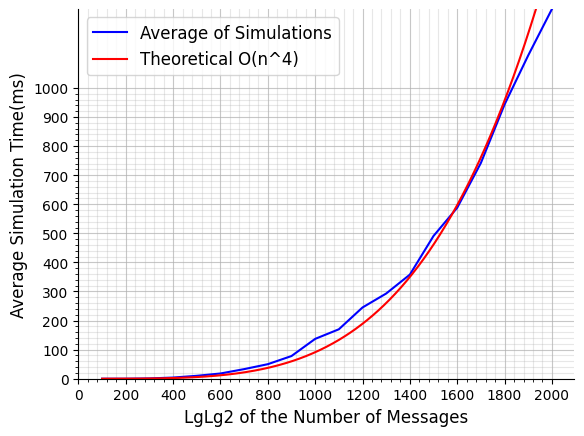}}
    \caption{Theoretical upper bound and empirical type II error probability corresponding to parameters \(\alpha = 1.5\) and \(R = 500\)}
\end{figure}

\section{Conclusions and Future Directions}
\label{Sec.Conclusions}
In this report, we optimized the 3-step algorithm in three ways. Firstly, by using pseudo prime generation algorithms such as Miller-Rabin's we have improved the running time of the key generation from exponential time to polynomial time in block length. Secondly, we have modified the coding scheme to send the generated prime keys instead of their index. Although this modification increases the block length, it does not affect the asymptotic rate. As a result, cost for the system design of the transmitter and receiver is simplified without affecting the optimality of the scheme. Finally, we used results from number theory to obtain better upper bounds on the type II error probability.

This work can be extended in several directions. The obvious way is to optimize some aspects of the method. For example, one might provide a new pseudo prime generation algorithm that requires less random bit or runs faster. A more interesting way is to investigates different numbers instead of primes. For instance, we can use semi-prime numbers, numbers with exactly two prime factors, as keys. This will increase the number of available keys for a given block length which will lead to overall shorter block length, however, its impact on the type II error probability is not known, yet. On the other hand, there exist other primality test in the literature which has a deterministic nature such as AKS \cite{Agrawal04} and can be examined as a possible extension of our work. Unlike the Miller Rabin test which generate only a probabilistic result, this algorithm can distinguish deterministically whether or not a target general number is prime or composite.

\section*{}
\bibliographystyle{IEEEtran}
\bibliography{Lit}

\end{document}

%% file: PKG.tex
\usepackage{enumitem}

\usepackage[most]{tcolorbox}
\tcbset{on line, 
        boxsep=4pt, left=0pt,right=0pt,top=0pt,bottom=0pt,
        colframe=white,colback=gray!10,  
        highlight math style={enhanced}
        }

\usepackage{amsbsy}
\usepackage{bm}
\usepackage{fixmath}
\usepackage{silence}
\usepackage{pst-node,pst-plot}
\usetikzlibrary{shapes.geometric}
\usepackage{mathdots}
\usetikzlibrary{math}

\usepackage{upgreek}

\usepackage{autobreak}
\allowdisplaybreaks

\usepackage[caption=false]{subfig}
\usepackage{commath,amsmath,amssymb,amsfonts}
\usepackage{mathtools}
\usepackage{amsthm}
\usepackage{pgfplots} 
\pgfplotsset{compat=1.15}
\usepackage{graphicx}
\usepackage{pgfgantt}
\usepackage{pdflscape}
\usepackage{pst-plot}
\usepackage{xfrac}
\usepackage{colortbl}
\usepackage{cancel}
\usepgfplotslibrary{fillbetween}
\usepackage{amssymb,bm}
\usepackage{float}
\usepackage{amsmath}
\usepackage[draft=false]{hyperref}
\usepackage{multirow}
\usepackage{xcolor}
\usepackage{mathrsfs}
\usepackage{bbm}

\newcommand{\floor}[1]{\left \lfloor #1 \right \rfloor}
\newcommand{\ceil}[1]{\left \lceil #1 \right \rceil}

\usepackage{cite} 

\usepackage{comment}

\usepackage{autobreak}
\allowdisplaybreaks

\def \C{\mathcal{C}}
\def \D{\mathcal{D}}

\def \K{\mathcal{K}}

\def \M{\mathcal{M}}

\def \P{\mathcal{P}}

\def \bB {\mathcal{B}}

\def \fc{\mathbf{c}}

\def \fy{\mathbf{y}}

\def \f0{\mathbf{0}}

\definecolor{blau_1a}{RGB}{93,133,195}
\definecolor{blau_2a}{RGB}{0,156,218}
\definecolor{gruen_3a}{RGB}{80,182,149}
\definecolor{gruen_4a}{RGB}{175,204,80}
\definecolor{gruen_5a}{RGB}{221,223,72}
\definecolor{orange_6a}{RGB}{255,224,92}
\definecolor{orange_7a}{RGB}{248,186,60}
\definecolor{rot_8a}{RGB}{238,122,52}
\definecolor{rot_9a}{RGB}{233,80,62}
\definecolor{lila_10a}{RGB}{201,48,142}
\definecolor{lila_11a}{RGB}{128,69,151}

\definecolor{blau_1b}{RGB}{0,90,169}
\definecolor{blau_2b}{RGB}{0,131,204}
\definecolor{gruen_3b}{RGB}{0,157,129}
\definecolor{gruen_4b}{RGB}{153,192,0}
\definecolor{gruen_5b}{RGB}{201,212,0}
\definecolor{orange_6b}{RGB}{253,202,0}
\definecolor{orange_7b}{RGB}{245,163,0}
\definecolor{rot_8b}{RGB}{236,101,0}
\definecolor{rot_9b}{RGB}{230,0,26}
\definecolor{lila_10b}{RGB}{166,0,132}
\definecolor{lila_11b}{RGB}{114,16,133}

\definecolor{mycolor1}{rgb}{0.0, 0.18, 0.39}
\definecolor{mycolor2}{RGB}{87,108,67}
\definecolor{mycolor3}{RGB}{8,133,161}
\definecolor{mycolor4}{RGB}{80,91,161}
\definecolor{mycolor5}{RGB}{98,122,157}
\definecolor{mycolor6}{RGB}{255,163,67}
\definecolor{mycolor7}{RGB}{152,205,225}
\definecolor{mycolor8}{RGB}{242,204,48}
\definecolor{mycolor9}{rgb}{0,.5,0}
\definecolor{mycolor10}{rgb}{.59,.44,.09}
\definecolor{mycolor11}{RGB}{231,199,31} 
\definecolor{mycolor12}{RGB}{8,133,161} 
\definecolor{mycolor13}{RGB}{157,188,64} 
\definecolor{mycolor14}{RGB}{194,150,130} 
\definecolor{mycolor15}{RGB}{98,122,157} 
\definecolor{mycolor16}{RGB}{160,160,160} 
\definecolor{mycolor17}{RGB}{115,82,68} 
\definecolor{mycolor18}{RGB}{94,60,108} 
\definecolor{mycolor19}{RGB}{115,82,68} 
\definecolor{mycolor20}{RGB}{255,183,30} 

\usepackage{accents}

\theoremstyle{remark} \newtheorem{theorem}{Theorem}
\theoremstyle{remark} \newtheorem{lemma}[theorem]{Lemma}
\theoremstyle{remark} 
\theoremstyle{remark} 
\theoremstyle{remark} \newtheorem{definition}{Definition}
\theoremstyle{remark} 
\theoremstyle{remark} \newtheorem{example}{Example}

\usepackage{amssymb}

\usepackage{tikz}
\usetikzlibrary{arrows.meta}

\usepackage{pst-node,pst-plot}

\usetikzlibrary{shapes.geometric}

\usetikzlibrary{arrows, arrows.meta}
\usepackage{pgfplots}
\pgfplotsset{compat=1.15}
\usepgfplotslibrary{polar}

\usepackage[bottom]{footmisc}

\usepackage{cleveref}
\crefname{equation}{Eq}{}
\crefrangelabelformat{equation}{(#3#1#4--#5#2#6)}

\usepackage{algpseudocode}
\usepackage[linesnumbered,ruled,noline]{algorithm2e}
\SetKwInput{KwInput}{Input}
\SetKwInput{KwOutput}{Output}
\SetKwBlock{Loop}{Loop}{end}

\usepackage{diffcoeff}
\usepackage{fullpage,amsmath,tkz-base}  
\usepackage{cases}

\usepackage{booktabs}

\usepackage{wrapfig}
\usepackage{upgreek}

\usepackage{siunitx}
\usepackage{framed}

%% file: Alg-1.tex
\begin{algorithm}[H]
\DontPrintSemicolon
\SetKwInOut{Input}{Input}\SetKwInOut{Output}{Output}
\Input{Positive Integer $n$}
\Output{Uniformly Selected Prime $p \leq n$}
\Repeat{$p$ is prime} 
{
    $p \gets \set{2,\dots, n};$
}
\Return{$p$}
\caption{Uniform Prime Generation}
\label{Alg.UPM-1}
\end{algorithm}

%% file: Alg-2.tex
\begin{algorithm}[H]
\DontPrintSemicolon
\SetKwInOut{Input}{Input}\SetKwInOut{Output}{Output}
\Input{Positive Integer $n$}
\Output{Uniformly Selected Prime $p \leq n$}
\Repeat{Until $p$ is probably a prime} 
{
    $p \gets \set{2,\dots, n};$
}
\Return{$p$}
\caption{Uniform Prime Generation}
\label{Alg.UPM-2}
\end{algorithm}

%% file: Alg-3.tex
\begin{algorithm}
	\DontPrintSemicolon
	\SetKwInOut{Input}{input}\SetKwInOut{Output}{output}
	\Input{positive integers \(N,s,k\)}
	\Output{A uniformly chosen prime number less than or equal to \(N\)}
	\For{\(i = 1 \to s\)}{
	\(n \gets \{1,2,\dots, N\}\)\;
	\If{\(MR(n,k)\)}{
		\Return{\(n\)}
	}
	}
	\Return{\(\perp\)}
	\caption{GMR(N,s,k)}
	\label{alg:GMR}
\end{algorithm}